\documentclass{mysvjour3}                     

\smartqed  

\usepackage{amssymb}
\usepackage{amsfonts}
\usepackage[colorlinks,citecolor=blue,urlcolor=blue,breaklinks=true]{hyperref}
\usepackage{cite}
\usepackage[cmex10]{amsmath}
\usepackage{url}

\journalname{Cryptogr. Commun.}

\begin{document}

\title{Combinatorial Bounds and Characterizations of\\ Splitting Authentication Codes
\thanks{This work was supported by the Deutsche Forschungsgemeinschaft (DFG) via a Heisenberg grant (Hu954/4) and a Heinz Maier-Leibnitz Prize grant (Hu954/5).}
}

\titlerunning{Splitting Authentication Codes}

\author{Michael Huber}

\authorrunning{Michael Huber}

\institute{M. Huber \at
              Wilhelm-Schickard-Institute for Computer Science, University of Tuebingen,\\
              Sand~13, 72076~Tuebingen, Germany\\
              \email{michael.huber@uni-tuebingen.de}           
          }

\date{Received: July 17, 2009 / Accepted: January 19, 2010}

\maketitle

\begin{abstract}
We present several generalizations of results for splitting authentication codes by studying the aspect of multi-fold security.
As the two primary results, we prove a combinatorial lower bound on the number of encoding rules and a combinatorial characterization of optimal splitting authentication codes that are multi-fold secure against spoofing attacks. The characterization is based on a new type of combinatorial designs, which we introduce and for which basic necessary conditions are given regarding their existence.
\keywords{Information authenticity \and unconditional security \and authentication code \and splitting \and non-deterministic encoding \and splitting design.}
\subclass{94A60 \and 94C30}
\end{abstract}

\section{Introduction}\label{Introduction}

Authenticity is one of the fundamental components in cryptography and information security.
Typically, communicating parties would like to be assured of the authenticity of information they obtain via potentially insecure
channels. Concerning unconditional (theoretical) authenticity, authentication codes can be used to minimize the possibility of an undetected deception. Their initial study appears to be that of Gilbert, MacWilliams \& Sloane~\cite{gil74}. A more general and systematic theory of authenticity was developed by Simmons\cite{Sim84,Sim85}.

We primarily focus on authentication codes with splitting in this paper. In such a code, several messages can be used to communicate a particular plaintext (\emph{non-deterministic encoding}). This concept plays an important role, for instance, in the context of authentication codes that permit arbitration (see, for example,~\cite{John94,Kur94,Kur01,Sim88,Sim90}).
We will deal with splitting authentication codes from a combinatorial point of view. By studying the aspect of multi-fold security, we obtain
several natural generalizations of results on splitting authentication codes.
As the two primary results, we prove a combinatorial lower bound on the number of encoding rules and a combinatorial characterization of optimal splitting authentication codes that are multi-fold secure against spoofing attacks.

For splitting authentication codes that are one-fold secure against spoofing attacks, Brickell~\cite{Brick84} and Simmons~\cite{Sim90} have established a combinatorial lower bound on the number of encoding rules. We will give a combinatorial lower bound accordingly for multi-fold secure splitting authentication codes.
Ogata \emph{et al.}~\cite{Ogata04} introduced splitting balanced incomplete block designs (BIBDs). They proved basic necessary conditions for their existence and derived a Fisher-type inequality. Furthermore, they established an equivalence between splitting BIBDs and optimal one-fold secure splitting authentication codes. We will extend the notion of splitting BIBDs to splitting \mbox{$t$-designs}. Comprehensive necessary conditions regarding their existence will be given. Moreover, we will prove an equivalence between splitting \mbox{$t$-designs} and optimal \mbox{$(t-1)$-fold} secure splitting authentication codes.

The paper is organized as follows. In Section~\ref{SplittingAuthenticationCodes}, we give the definition and concept of multi-fold secure splitting authentication codes. We introduce splitting designs and prove basic necessary conditions for their existence in Section~\ref{SplittingDesigns}.
With respect to our further purposes, we summarize in Section~\ref{Authentication Codes without Splitting} the state-of-the-art for authentication codes without splitting. In Section~\ref{CombinatorialBoundsforSplittingAuthenticationCodes}, lower bounds on deception probabilities and the number of encoding rules are established for multi-fold secure splitting authentication codes.
A combinatorial characterization of optimal multi-fold secure splitting authentication codes in terms of splitting designs is given in Section~\ref{CombinatorialCharacterizationsofSplittingAuthenticationCodes}. We finally conclude the paper and propose further research problems in Section~\ref{conclusion}.


\section{Splitting Authentication Codes}\label{SplittingAuthenticationCodes}

Splitting authentication codes were first introduced by Simmons~\cite{Sim82}. These codes are useful, inter alia, for the analysis of authentication codes with arbitration. In particular,~\cite{Kur01} gives an equivalence between splitting authentication codes and authentication codes with arbitration.

We use the unconditional (theoretical) secure authentication model developed by Simmons (e.g.~\cite{Sim82,Sim84,Sim85,Sim92}). Our notation follows, for the most part, that of~\cite{Mass86,Ogata04,Stin90}.
In this model, three participants are involved: a \emph{transmitter}, a \emph{receiver}, and an \emph{opponent}.  The transmitter wants to communicate information to the receiver via a public communications channel. The receiver in return would like to be confident that any received information actually came from the transmitter and not from some opponent (\emph{integrity} of information). The transmitter and the receiver are assumed to trust each other. Sometimes this is also called an \emph{$A$-code}.

Let $\mathcal{S}$ denote a finite set of \emph{source states} (or \emph{plaintexts}), $\mathcal{M}$ a finite set of \emph{messages} (or \emph{ciphertexts}), and $\mathcal{E}$ a finite set of \emph{encoding rules} (or \emph{keys}). Using an encoding rule $e\in \mathcal{E}$,
the transmitter encrypts a source state $s \in \mathcal{S}$ to obtain the message $m=e(s)$ to be sent over the channel.
The encoding rule is communicated to the receiver via a secure channel prior to any messages being sent.
When it is possible that more than one message can be used to communicate a particular source state $s \in \mathcal{S}$ under the same encoding rule $e\in \mathcal{E}$, then the authentication code is said to have \emph{splitting}. In this case, a message $m \in \mathcal{M}$ is computed as $m=e(s,r)$, where $r$ denotes a random number chosen from some specified finite set $\mathcal{R}$. If we define
\[e(s):=\{m \in \mathcal{M}: m = e(s,r)\;\, \mbox{for some}\;\, r \in \mathcal{R}\}\]
for each encoding rule $e\in \mathcal{E}$ and each source state $s \in \mathcal{S}$,
then splitting means that $\left|e(s)\right|>1$ for some $e\in \mathcal{E}$ and some $s \in \mathcal{S}$.
In order to ensure that the receiver can decrypt the message being sent, it is required for any $e \in \mathcal{E}$ that $e(s) \cap e(s^\prime)=\emptyset$ if $s \neq s^\prime$.
For a given encoding rule $e \in \mathcal{E}$, let \[M(e):= \bigcup_{s \in \mathcal{S}} e(s)\] denote the set of \emph{valid} messages. For an encoding rule $e$ and a set $M^\prime \subseteq M(e)$ of distinct messages, we define
\[f_e(M^\prime):=\{s \in \mathcal{S}: e(s) \cap M^\prime \neq \emptyset\},\]
i.e., the set of source states that will be encoded under encoding rule $e$ by a message in $M^\prime$.
A received message $m$ will be accepted by the receiver as being authentic if and only if $m \in M(e)$. When this is fulfilled, the receiver decrypts the message $m$ by applying the decoding rule $e^{-1}$, where \[e^{-1}(m)=s\;\, \mbox{if}\;\,  m = e(s,r)\;\, \mbox{for some}\;\, r \in \mathcal{R}.\]
A splitting authentication code is called \emph{$c$-splitting} if
\[\left|e(s)\right|=c\]
for every encoding rule $e \in \mathcal{E}$ and every source state \mbox{$s \in \mathcal{S}$}.
We note that an authentication code can be represented algebraically by a $(\left|\mathcal{E}\right| \times \left|\mathcal{S}\right|)$\emph{-encoding matrix} with the rows indexed by the encoding rules $e \in \mathcal{E}$, the columns indexed by the source states $s \in \mathcal{S}$, and the entries defined by $a_{es}:=e(s)$. As a simple example, Table~\ref{encod_gen} displays an encoding matrix of
a \mbox{$2$-splitting} authentication code for $2$ source states, having $9$ messages and $9$ encoding rules (cf.~Example~\ref{example_concrete}).

\begin{table}
\renewcommand{\arraystretch}{1.3}
\caption{An example of a splitting authentication code.}\label{encod_gen}

\hspace*{0.1cm}

\begin{center}

\begin{tabular}{|c| c c|}
  \hline
  & $s_1$ & $s_2$ \\
  \hline
  $e_1$ & \{$m_1,m_2$\} & \{$m_3,m_5$\}\\
  $e_2$ & \{$m_2,m_3$\} & \{$m_4,m_6$\}\\
  $e_3$ & \{$m_3,m_4$\} & \{$m_5,m_7$\}\\
  $e_4$ & \{$m_4,m_5$\} & \{$m_6,m_8$\}\\
  $e_5$ & \{$m_5,m_6$\} & \{$m_7,m_9$\}\\
  $e_6$ & \{$m_6,m_7$\} & \{$m_8,m_1$\}\\
  $e_7$ & \{$m_7,m_8$\} & \{$m_9,m_2$\}\\
  $e_8$ & \{$m_8,m_9$\} & \{$m_1,m_3$\}\\
  $e_9$ & \{$m_9,m_1$\} & \{$m_2,m_4$\}\\
  \hline
\end{tabular}

\end{center}
\end{table}

\subsection{Protection Against Spoofing Attacks}

We address the scenario of a \emph{spoofing attack} of order $i$ (cf.~\cite{Mass86}):
Suppose that an opponent observes $i\geq 0$ distinct messages, which are sent through the public channel using the same encoding rule. The opponent then inserts a new message $m^\prime$ (being distinct from the $i$ messages already sent), hoping to have it accepted by the receiver as authentic.
The cases $i=0$ and $i=1$ are called \emph{impersonation game} and \emph{substitution game}, respectively. These cases have been studied in detail in recent years for authentication codes without splitting (see, e.g.,~\cite{Stin92,Stin96}) and with splitting (see, e.g.,~\cite{Blund99,DeSoete91,Ogata04}). However, much less is known for the cases $i \geq 2$ in particular for splitting authentication codes.

For any $i$, we assume that there is some probability distribution on the set of \mbox{$i$-subsets} of source states, so that any set of $i$ source states has a non-zero probability of occurring. For simplification, we ignore the order in which the $i$ source states occur, and assume that no source state occurs more than once.
Given this probability distribution on the set $\mathcal{S}$ of source states, the receiver and transmitter also choose a probability distribution on the set $\mathcal{E}$ of encoding rules, called an \emph{encoding strategy}. It is assumed that the opponent knows the encoding strategy being used. If splitting occurs, then the receiver/transmitter will also choose a \emph{splitting strategy} to determine $m \in \mathcal{M}$, given $s \in \mathcal{S}$ and $e \in \mathcal{E}$ (this corresponds to \emph{non-deterministic encoding}). The transmitter/receiver will determine these strategies to minimize the chance of being deceived by the opponent.
The \emph{deception probability} $P_{d_i}$ denotes the probability that the opponent can deceive the transmitter/receiver with a spoofing attack of order $i$.


\section{Splitting Designs}\label{SplittingDesigns}

There are natural and deep connections between authentication codes and combinatorial designs, see, for example,~\cite{crc06,Hub2009,Ogata04,pei06,Stin90,Stin92,Stin96}. The close relationship between
cryptography and designs was presumably first revealed in Shannon's classical paper~\cite{Shan49}
on secrecy systems.

In order to give a combinatorial characterization of multi-fold secure splitting authentication codes in the remainder of the paper,
we define in this section a new type of combinatorial designs. Let us first
recall the classical notion of a combinatorial \mbox{$t$-design} (see, for instance,~\cite{BJL1999}):

\begin{definition}
For positive integers $t \leq k \leq v$ and $\lambda$, a
\mbox{\emph{$t$-$(v,k,\lambda)$ design}} $\mathcal{D}$ is a pair \mbox{$(X,\mathcal{B})$}, satisfying
the following properties:

\begin{enumerate}

\item[(i)] $X$ is a set of $v$ elements, called \emph{points},

\smallskip

\item[(ii)] $\mathcal{B}$ is a family of \mbox{$k$-subsets} of $X$, called \emph{blocks},

\smallskip

\item[(iii)] every \mbox{$t$-subset} of $X$ is contained in exactly $\lambda$ blocks.

\end{enumerate}
\end{definition}

By convention, $b:=\left| \mathcal{B} \right|$ denotes the number of blocks. It is easily seen that
\[b=\lambda {v \choose t}/ {k \choose t}.\]

For encyclopedic references on combinatorial \mbox{$t$-designs}, we refer to~\cite{BJL1999,crc06}. A recent treatment on highly regular designs and their applications in information and coding theory can be found, e.g., in~\cite{Hu2008,Hu2009}.

The notion of a splitting balanced incomplete block design (BIBD) have been introduced by Ogata \emph{et al.}~\cite{Ogata04}.
We will extend this concept to \emph{splitting \mbox{$t$-designs}}:

\begin{definition}\label{splittdesign}
For positive integers $t,v,b,c,u,\lambda$ with \mbox{$t \leq u$} and \mbox{$cu \leq v$}, a
\mbox{$t$-$(v,b,l=cu,\lambda)$} \emph{splitting design} $\mathcal{D}$ is a pair \mbox{$(X,\mathcal{B})$}, satisfying
the following properties:

\begin{enumerate}

\item[(i)] $X$ is a set of $v$ elements, called \emph{points},

\smallskip

\item[(ii)] $\mathcal{B}$ is a family of \mbox{$l$-subsets} of $X$, called \emph{blocks}, such that every block $B_i \in \mathcal{B}$
$(1\leq i \leq \left|\mathcal{B}\right|=:b)$ is expressed as a disjoint union \[B_i= B_{i,1} \cup \cdots \cup B_{i,u}\] with $\left|B_{i,1}\right| = \cdots = \left|B_{i,u}\right|=c$ and $\left|B_i\right|=l=cu$,

\smallskip

\item[(iii)]  every \mbox{$t$-subset} $\{x_m\}_{m=1}^t $of $X$ is contained in exactly $\lambda$ blocks $B_i=B_{i,1} \cup \cdots \cup B_{i,u}$ such that
\[x_m \in B_{i,j_m} \quad (j_m \; \mbox{between} \;  1 \; \mbox{and} \; u)\]
for each $1\leq m \leq t$,  and $j_1, \ldots, j_t$ are mutually distinct.
\end{enumerate}
\end{definition}

\begin{example}\label{example}
A splitting \mbox{$2$-design} is a splitting BIBD. As a simple example, take as point set \[X =\{1,2,3,4,5,6,7,8,9\}\]
and as block set \[\mathcal{B}=\{B_1,\ldots,B_9\}\] with
\begin{align*}{}
B_1 &=\{\{1,2\},\, \{3,5\}\}\\
B_2 &=\{\{2,3\},\, \{4,6\}\}\\
B_3 &=\{\{3,4\},\, \{5,7\}\}\\
B_4 &=\{\{4,5\},\, \{6,8\}\}\\
B_5 &=\{\{5,6\},\, \{7,9\}\}\\
B_6 &=\{\{6,7\},\, \{8,1\}\}\\
B_7 &=\{\{7,8\},\, \{9,2\}\}\\
B_8 &=\{\{8,9\},\, \{1,3\}\}\\
B_9 &=\{\{9,1\},\, \{2,4\}\}.
\end{align*}
This gives a \mbox{$2$-$(9,9,4=2\times2,1)$} splitting design (see~\cite[Ex.\,5.1]{Ogata04}).
\end{example}

\begin{example}\label{example_Hu}
A \mbox{$3$-$(10,15,6=2\times3,1)$} splitting design can be obtained (via a computer search) by taking
as point set \[X =\{1,2,3,4,5,6,7,8,9,0\}\]
and as block set \[\mathcal{B}=\{B_1,\ldots,B_{15}\}\] with
\begin{align*}{}
B_1 &=\{\{1,2\},\, \{4,0\},\, \{5,9\}\}\\
B_2 &=\{\{1,3\},\, \{2,8\},\, \{5,0\}\}\\
B_3 &=\{\{1,4\},\, \{3,8\},\, \{6,9\}\}\\
B_4 &=\{\{1,5\},\, \{4,7\},\, \{6,8\}\}\\
B_5 &=\{\{1,7\},\, \{2,3\},\, \{4,8\}\}\\
B_6 &=\{\{1,8\},\, \{2,5\},\, \{6,9\}\}\\
B_7 &=\{\{1,8\},\, \{6,7\},\, \{9,0\}\}\\
B_8 &=\{\{1,9\},\, \{2,5\},\, \{3,7\}\}\\
B_9 &=\{\{1,9\},\, \{3,4\},\, \{7,0\}\}\\
B_{10} &=\{\{2,4\},\, \{5,6\},\, \{7,9\}\}\\
B_{11} &=\{\{2,5\},\, \{4,7\},\, \{3,0\}\}\\
B_{12} &=\{\{2,9\},\, \{6,8\},\, \{3,0\}\}\\
B_{13} &=\{\{2,0\},\, \{4,5\},\, \{6,8\}\}\\
B_{14} &=\{\{3,7\},\, \{4,6\},\, \{8,0\}\}\\
B_{15} &=\{\{3,9\},\, \{5,7\},\, \{6,0\}\}.
\end{align*}
\end{example}

We prove some basic necessary conditions for the existence of splitting designs:

\begin{proposition}\label{s-design_splitt}
Let $\mathcal{D}=(X,\mathcal{B})$ be a \mbox{$t$-$(v,b,l=cu,\lambda)$} splitting design, and for a positive integer $s \leq t$, let $S \subseteq X$
with $\left|S\right|=s$. Then the number of blocks containing
each element of $S$ as per Definition~\ref{splittdesign} is given by
\[\lambda_s = \lambda \frac{{v-s \choose t-s}}{c^{t-s}{u-s \choose t-s}}.\]
In particular, for $t\geq 2$, a \mbox{$t$-$(v,b,l=cu,\lambda)$} splitting design is
also an \mbox{$s$-$(v,b,l=cu,\lambda_s)$} splitting design.
\end{proposition}

\begin{proof}
We count in two ways the number of pairs $(T,B_i)$, where $T:=\{x_m\}_{m=1}^t\subseteq
X$ and $B_i = \bigcup_{j=1}^u B_{i,j} \in \mathcal{B}$ such that
\[x_m \in B_{i,j_m}\]
for each $1\leq m \leq t$ with $j_1, \ldots, j_t$ mutually distinct, and \mbox{$S:=\{\tilde{x}_m\}_{m=1}^s \subseteq T$}. First, each of the $\lambda_s$ blocks $B_i=\bigcup_{j=1}^u B_{i,j}$ such that
\[\tilde{x}_m \in B_{i,j_m}\]
for each $1\leq m \leq s$ with $j_1, \ldots, j_s$ mutually distinct gives
\[\frac{\prod_{i=s}^{t-1} (l-ic)}{(t-s)!}=c^{t-s}{u-s \choose t-s}\]
such pairs. Second, there
are \[{v-s \choose t-s}\] such subsets $T \subseteq X$ with $S \subseteq T$, each giving $\lambda$ pairs by
Definition~\ref{splittdesign}.\qed
\end{proof}

As it is customary for \mbox{$t$-designs}, we also set $r:= \lambda_1$ denoting the
number of blocks containing a given point. The above elementary counting arguments give the following
assertions.

\begin{proposition}\label{Comb_t=5_splitt}
Let $\mathcal{D}=(X,\mathcal{B})$ be a \mbox{$t$-$(v,b,l=cu,\lambda)$} splitting design. Then the following holds:

\begin{enumerate}

\item[\emph{(a)}] $bl = vr.$

\smallskip

\item[\emph{(b)}] $\displaystyle{{v \choose t} \lambda = bc^t {u \choose t}.}$

\smallskip

\item[\emph{(c)}] $rc^{t-1}(u-1)=\lambda_2(v-1)$ for $t \geq 2$.

\end{enumerate}
\end{proposition}

\begin{remark}
The above proposition extends the result~\cite[Lemma\,5.1]{Ogata04},
where (b) and (c) have been proved for the case when \mbox{$t=2$}.
\end{remark}

Since in Proposition~\ref{s-design_splitt} each $\lambda_s$ must be an integer,
we obtain furthermore the subsequent necessary arithmetic conditions.

\begin{proposition}\label{divCond}
Let $\mathcal{D}=(X,\mathcal{B})$ be a \mbox{$t$-$(v,b,l=cu,\lambda)$} splitting design.
Then
\[\lambda {v-s \choose t-s} \equiv  0 \, \bigg(\hspace{-0.3cm} \mod c^{t-s}{u-s \choose t-s}\bigg)\]
for each positive integer $s \leq t$.
\end{proposition}

Ogata \emph{et al.}~\cite{Ogata04} proved a Fisher-type inequality for splitting BIBDs. As a splitting \mbox{$t$-design} with $t \geq 2$ is also a splitting \mbox{$2$-design} in view of Lemma~\ref{s-design_splitt}, we derive

\begin{proposition}\label{FisherIn_splitt}
If $\mathcal{D}=(X,\mathcal{B})$ is a \mbox{$t$-$(v,b,l=cu,\lambda)$} splitting design with $t \geq 2$, then
\[b\geq \frac{v}{u}.\]
\end{proposition}


\section{Authentication Codes without Splitting}\label{Authentication Codes without Splitting}

With respect to our further purposes, we summarize the state-of-the-art for authentication codes without splitting:

The following theorems (cf.~\cite{Mass86,Sch86}) give combinatorial lower bounds on cheating probabilities as well as on the size of encoding rules for multi-fold secure authentication codes:

\begin{theorem}[Massey]
In an authentication code without splitting, for every $0 \leq i \leq t$, the deception probabilities are bounded below by
\[P_{d_i}\geq \frac{\left|\mathcal{S}\right|-i}{\left|\mathcal{M}\right|-i}.\]
\end{theorem}

We remark that a code is called $t$\emph{-fold secure against spoofing} if
\[P_{d_i}= (\left|\mathcal{S}\right|-i)/(\left|\mathcal{M}\right|-i)\]
for all $0 \leq i \leq t$.

\begin{theorem}[Massey--Sch\"{o}bi]
If an authentication code without splitting is \mbox{$(t-1)$-fold} secure against spoofing, then the number of encoding rules is bounded below by
\[\left|\mathcal{E}\right| \geq \frac{{\left|\mathcal{M}\right| \choose t}}{{\left|\mathcal{S}\right| \choose t}}.\]
\end{theorem}

Such a code is called \emph{optimal} if the number of encoding rules meets the lower bound with equality. When the source states are known to be independent and equiprobable, optimal authentication codes without splitting which are multi-fold secure against spoofing have been characterized via \mbox{$t$-designs} (cf.~\cite{DeS88,Sch86,Stin90}).

\begin{theorem}[DeSoete--Sch\"{o}bi--Stinson]\label{general}
Suppose there is a \mbox{$t$-$(v,k,\lambda)$} design. Then there is an authentication code without splitting for $k$ equiprobable source states, having $v$ messages and $\lambda {v \choose t}/{k \choose t}$ encoding rules, that is $(t-1)$-fold secure against spoofing. Conversely, if there is an optimal authentication code without splitting for $k$ equiprobable source states, having  $v$ messages and ${v \choose t}/{k \choose t}$ encoding rules, that is $(t-1)$-fold secure against spoofing, then there is a \mbox{$t$-$(v,k,1)$} design.
\end{theorem}

Combinatorial constructions of optimal multi-fold secure authentication codes without splitting which simultaneously achieve perfect \mbox{secrecy} have been obtained recently via the following theorem (see~\cite{Hub2009}).

\begin{theorem}[Huber]\label{mythm1_ISIT}
Suppose there is a \mbox{$t$-$(v,k,1)$} design, where $v$ divides the number of blocks $b$.
Then there is an optimal authentication code without splitting for $k$ equiprobable source states, having $v$ messages and ${v \choose t}/{k \choose t}$ encoding rules, that is $(t-1)$-fold secure against spoofing and provides perfect secrecy.
\end{theorem}


\section{Combinatorial Bounds for Splitting Authentication Codes}\label{CombinatorialBoundsforSplittingAuthenticationCodes}

In this section, we give combinatorial lower bounds on deception probabilities, and a combinatorial lower bound on the size of encoding rules for splitting authentication codes that are multi-fold secure against spoofing.

We first state lower bounds on cheating probabilities for splitting authentication codes (see~\cite{Blund99,DeSoete91}).

\begin{theorem}[DeSoete--Blundo--DeSantis--Kurosawa--Ogata]\label{deceptprob_splitt}
In a splitting authentication code, for every $0 \leq i \leq t$, the deception probabilities are bounded below by
\[P_{d_i}\geq \min_{e \in \mathcal{E}} \frac{\left|M(e)\right| - i \cdot \max_{s \in \mathcal{S}} \left| e(s) \right|}{\left|\mathcal{M}\right|-i}.\]
\end{theorem}

A splitting authentication code is called $t$\emph{-fold secure against spoofing} if
\[P_{d_i}= \min_{e \in \mathcal{E}} \frac{\left|M(e)\right| - i \cdot \max_{s \in \mathcal{S}} \left| e(s) \right|}{\left|\mathcal{M}\right|-i}\]
for all $0 \leq i \leq t$.

\smallskip

We prove now a lower bound on the size of encoding rules for multi-fold secure splitting authentication codes.

\begin{theorem}\label{my1}
If a splitting authentication code is \mbox{$(t-1)$-fold} secure against spoofing, then the number of encoding rules is bounded below by
\[\left|\mathcal{E}\right|  \geq \prod_{i=0}^{t-1} \frac{\left|\mathcal{M}\right|  -i}{\left|M(e)\right| - i \cdot \max_{s \in \mathcal{S}}\left| e(s) \right|}.\]
\end{theorem}

\begin{proof}
Let $M^{\prime} \subseteq \mathcal{M}$ be a set of $i \leq t-1$ distinct messages that are valid under a particular encoding rule, in such a way that they define $i$ different source states.
Let $x\in \mathcal{M}$ be any message not in $M^{\prime}$. We assume that there is no encoding rule $e \in \mathcal{E}$ under which all messages in $M^{\prime} \cup \{x\}$ are valid and for which $f_e(x) \notin f_e(M^{\prime})$.
Following the proof of Theorem~\ref{deceptprob_splitt}, mutatis mutandis, yields
\[P_{d_i} > \min_{e \in \mathcal{E}} \frac{\left|M(e)\right| - i \cdot \max_{s \in \mathcal{S}} \left| e(s) \right|}{\left|\mathcal{M}\right|-i},\]
a contradiction. Therefore, any set of $t$ distinct messages is valid under at least one encoding rule such that they define different source states. The bound follows now by counting in two ways the number of \mbox{$t$-subsets} of messages that are valid under some encoding rule such that they correspond to different source states.\qed
\end{proof}

Analogously, we call a splitting authentication code \emph{optimal} if the number of encoding rules meets the lower bound with equality.

\begin{remark}
The above theorem generalizes results by Brickell~\cite{Brick84} and Simmons~\cite{Sim90},
where a lower bound in the case of one-fold secure splitting authentication codes has been established.
\end{remark}

As a consequence, we obtain for \mbox{$c$-splitting} authentication codes the following lower bounds:

\begin{corollary}\label{Cor1}
In a $c$-splitting authentication code,
\[P_{d_i}\geq \frac{c(\left|\mathcal{S}\right|-i)}{\left|\mathcal{M}\right|-i}\]
for every $0 \leq i \leq t$.
\end{corollary}

\begin{proof}
We set $l:=\left|M(e)\right|=c\left|\mathcal{S}\right|$. Then Theorem~\ref{deceptprob_splitt} yields
\[P_{d_i}\geq \frac{l - i \cdot c}{\left|\mathcal{M}\right|-i}= \frac{c(\left|\mathcal{S}\right|-i)}{\left|\mathcal{M}\right|-i}.\]\qed
\end{proof}

\begin{corollary}\label{Cor2}
If a $c$-splitting authentication code is \mbox{$(t-1)$}-fold secure against spoofing, then
\[\left|\mathcal{E}\right|  \geq \frac{{\left|\mathcal{M}\right| \choose t}}{c^t{\left|\mathcal{S}\right| \choose t}}.\]
\end{corollary}

\begin{proof}
Using Theorem~\ref{my1}, we may proceed as for Corollary~\ref{Cor1}.\qed
\end{proof}


\section{Combinatorial Characterizations of Optimal Splitting Authentication Codes}\label{CombinatorialCharacterizationsofSplittingAuthenticationCodes}

Ogata \emph{et al.}~\cite[Thms.\,5.4\,and\,5.5]{Ogata04} characterized in 2004 optimal splitting authentication codes that are one-fold secure against spoofing. Their combinatorial result is based on splitting BIBDs.

\begin{theorem}[Ogata--Kurosawa--Stinson--Saido]\label{ogata_char}
Suppose there is a \mbox{$2$-$(v,b,l=cu,1)$} splitting design. Then there is an optimal \mbox{$c$-splitting} authentication code for $u$ equiprobable source states, having $v$ messages and ${v \choose 2}/[c^2{u \choose 2}]$ encoding rules, that is one-fold secure against spoofing. Conversely, if there is an optimal \mbox{$c$-splitting} authentication code for $u$ source states, having  $v$ messages and ${v \choose 2}/[c^2{u \choose 2}]$ encoding rules, that is one-fold secure against spoofing, then there is a \mbox{$2$-$(v,b,l=cu,1)$} splitting design.
\end{theorem}

An example is as follows (cf.~\cite[Ex.\,5.2]{Ogata04}):

\begin{example}\label{example_concrete}
An optimal \mbox{$2$-splitting} authentication code for $u=2$ equiprobable source states, having $v=9$ messages and $b=9$ encoding rules, that is one-fold secure against spoofing can be constructed from the \mbox{$2$-$(9,9,4=2\times2,1)$} splitting design in Example~\ref{example}. Each encoding rule is used with probability $1/9$. An encoding matrix is given in Table~\ref{encod}.

\begin{table}
\renewcommand{\arraystretch}{1.3}
\caption{Splitting authentication code from a \mbox{$2$-$(9,9,4=2\times2,1)$} splitting design.}\label{encod}

\hspace*{0.1cm}

\begin{center}

\begin{tabular}{|c| c c|}
  \hline
  & $s_1$ & $s_2$ \\
  \hline
  $e_1$ & \{1,2\} & \{3,5\}\\
  $e_2$ & \{2,3\} & \{4,6\}\\
  $e_3$ & \{3,4\} & \{5,7\}\\
  $e_4$ & \{4,5\} & \{6,8\}\\
  $e_5$ & \{5,6\} & \{7,9\}\\
  $e_6$ & \{6,7\} & \{8,1\}\\
  $e_7$ & \{7,8\} & \{9,2\}\\
  $e_8$ & \{8,9\} & \{1,3\}\\
  $e_9$ & \{9,1\} & \{2,4\}\\
  \hline
\end{tabular}

\end{center}
\end{table}

\end{example}

We give a natural extension of Theorem~\ref{ogata_char}. We prove that optimal splitting authentication codes that are multi-fold secure against spoofing can be characterized in terms of splitting \mbox{$t$-designs}.

\begin{theorem}\label{my2}
Suppose there is a \mbox{$t$-$(v,b,l=cu,1)$} splitting design with $t \geq 2$. Then there is an optimal \mbox{$c$-splitting} authentication code for $u$ equiprobable source states, having $v$ messages and ${v \choose t}/[c^t{u \choose t}]$ encoding rules, that is $(t-1)$-fold secure against spoofing. Conversely, if there is an optimal \mbox{$c$-splitting} authentication code for $u$ source states, having  $v$ messages and ${v \choose t}/[c^t{u \choose t}]$ encoding rules, that is \mbox{$(t-1)$-fold} secure against spoofing, then there is a \mbox{$t$-$(v,b,l=cu,1)$} splitting design.
\end{theorem}

\begin{proof}
Let us first assume that there is an optimal \mbox{$c$-splitting} authentication code for $\left|\mathcal{S}\right|:=u$ source states, having  $\left|\mathcal{M}\right|:=v$ messages and $\left|\mathcal{E}\right|:={v \choose t}/[c^t{u \choose t}]$ encoding rules, that is \mbox{$(t-1)$-fold} secure against spoofing. In order to meet the lower bound in Theorem~\ref{my1} with equality, every set of $t$ distinct messages must be valid under precisely one encoding rule, in such a way that they define different source states. For $e \in \mathcal{E}$, let us define a block $B_e \in \mathcal{B}$ as disjoint union
\[B_e := \bigcup_{j=1}^u e(s_j).\]
Then $(X,\mathcal{B})=(\mathcal{M},\{B_e:e\in \mathcal{E}\})$ is a \mbox{$t$-$(v,b,l=cu,1)$} splitting design in view of Definition~\ref{splittdesign}.

To prove the other direction, let $\left|\mathcal{M}\right|:=v$. For every block $B_i \in \mathcal{B}$ with
\[B_i= \bigcup_{j=1}^u B_{i,j},\]
we arbitrarily define an encoding rule $e_{B_i}$ via
\[e_{B_i}(s_j):=B_{i,j}\]
for each $1 \leq j \leq u$. Using every encoding rule with equal probability $1/b$ establishes the claim.\qed
\end{proof}

We present an example:

\begin{example}
An optimal \mbox{$2$-splitting} authentication code for $u=3$ equiprobable source states, having $v=10$ messages and $b=15$ encoding rules, that is two-fold secure against spoofing can be constructed from the \mbox{$3$-$(10,15,6=2\times3,1)$} splitting design in Example~\ref{example_Hu}. Each encoding rule is used with probability $1/15$. An encoding matrix is given in Table~\ref{encod_Hu}.

\begin{table}
\renewcommand{\arraystretch}{1.3}
\caption{Splitting authentication code from a \mbox{$3$-$(10,15,6=2\times3,1)$} splitting design.}\label{encod_Hu}

\hspace*{0.1cm}

\begin{center}

\begin{tabular}{|c| c c c|}
  \hline
  & $s_1$ & $s_2$ & $s_3$\\
  \hline
$e_1$ &\{1,2\} &\{4,0\} &\{5,9\}\\
$e_2$ &\{1,3\} &\{2,8\} &\{5,0\}\\
$e_3$ &\{1,4\} &\{3,8\} &\{6,9\}\\
$e_4$ &\{1,5\} &\{4,7\} &\{6,8\}\\
$e_5$ &\{1,7\} &\{2,3\} &\{4,8\}\\
$e_6$ &\{1,8\} &\{2,5\} &\{6,9\}\\
$e_7$ &\{1,8\} &\{6,7\} &\{9,0\}\\
$e_8$ &\{1,9\} &\{2,5\} &\{3,7\}\\
$e_9$ &\{1,9\} &\{3,4\} &\{7,0\}\\
$e_{10}$ &\{2,4\} &\{5,6\} &\{7,9\}\\
$e_{11}$ &\{2,5\} &\{4,7\} &\{3,0\}\\
$e_{12}$ &\{2,9\} &\{6,8\} &\{3,0\}\\
$e_{13}$ &\{2,0\} &\{4,5\} &\{6,8\}\\
$e_{14}$ &\{3,7\} &\{4,6\} &\{8,0\}\\
$e_{15}$ &\{3,9\} &\{5,7\} &\{6,0\}\\
\hline
\end{tabular}

\end{center}
\end{table}

\end{example}


\section{Conclusion}\label{conclusion}

We have given a combinatorial lower bound on the number of encoding rules for splitting authentication codes that are multi-fold secure against spoofing attacks. Moreover, we have provided a combinatorial characterization of those codes that attain these bounds. Our characterization was based on a new type of combinatorial designs, which we introduced and for which basic necessary conditions regarding their existence were given.
For future research, at least two directions would be of interest:

\begin{enumerate}

\item[(i)] Construction of multi-fold secure splitting authentication codes: Using Theorem~\ref{my2}, this asks for constructing \mbox{$t$-$(v,b,l=cu,1)$} splitting design for $t>2$. We remark that in the case when $t=2$, various combinatorial constructions have been obtained recently in~\cite{Du04,Ge05,Wan06} via recursive and direct constructions by the method of differences.

\smallskip

\item[(ii)] Including the aspect of perfect secrecy: Is it possible to give a characterization of optimal splitting authentication codes that are multi-fold secure against spoofing and simultaneously achieve perfect \mbox{secrecy} in the sense of Shannon? For the case of multi-fold secure authentication codes without splitting such a result has been established lately in~\cite{Hub2009} (cf.~Theorem~\ref{mythm1_ISIT}).

\end{enumerate}

\section*{Acknowledgments}
I thank the two anonymous referees for their careful reading and suggestions that helped improving the presentation of the paper. I also thank
Moritz Eilers and Christoff Hische for running the computer search for Example~\ref{example_Hu}.


\end{document}